\newif\iflong
\newif\ifshort

\iflong 
\else
\shorttrue
\fi

\documentclass{article}

\usepackage{balance} 

\newcommand{\mytitle}{Popularity and Perfectness in One-sided Matching Markets with Capacities}

\title{\mytitle}
\author{Gergely Cs{\'a}ji\textsuperscript{\rm 1}}

\date{\
\textsuperscript{\rm 1} \small HUN-REN Centre for Economic and Regional Studies, Hungary 
}
 
\usepackage[utf8]{inputenc}
\usepackage{amsmath,amsthm,amssymb}
\usepackage{graphicx,dsfont,mathtools}

\usepackage{xspace}

\usepackage{hyperref}
\hypersetup{%
 backref=true, 
 pagebackref=true, 
 hypertexnames=true,
 colorlinks=true,citecolor=green!35!black,linkcolor=red!60!black%
} 

\usepackage{xifthen}
\usepackage{paralist}

\usepackage{tikz}
\usetikzlibrary{decorations,arrows,arrows.meta,petri,topaths,backgrounds,shapes,positioning,fit,calc,decorations.pathreplacing,patterns,intersections}

\pgfdeclarelayer{bg}    
\pgfdeclarelayer{fg}    
\pgfsetlayers{bg,main,fg}  

\usepackage[textsize=tiny,textwidth=1.5cm,linecolor=green!70!black, backgroundcolor=green!10, bordercolor=black]{todonotes}

\usepackage[noend,ruled,linesnumbered]{algorithm2e}

\SetAlFnt{\small}
\SetAlCapFnt{\small}
\SetAlCapNameFnt{\small} 
\SetAlCapHSkip{0pt}

\SetKw{KwAnd}{and}
\SetKw{KwOr}{or}
\SetKw{KwNo}{no}
\SetKw{KwYes}{yes}
\SetKw{KwT}{true}
\SetKw{KwF}{false}
\SetKw{KwST}{s.t.}

\SetKwComment{Mycomment}{$\triangleright$\ }{}

\usepackage[capitalize]{cleveref}

\newtheorem{lemma}{Lemma}

\newtheorem{example}{Example}

\newtheorem{theorem}{Theorem}
\newtheorem{claim}{Claim}[theorem]

\newtheorem{corollary}{Corollary}
 \theoremstyle{definition}

\crefname{table}{Table}{Tables}
\crefname{figure}{Figure}{Figures}
\crefname{theorem}{Theorem}{Theorems}
\crefname{corollary}{Corollary}{Corollaries}
\crefname{observation}{Observation}{Observations}
\crefname{lemma}{Lemma}{Lemmas}
\crefname{example}{Example}{Examples}
\crefname{reduction}{Reduction}{Reductions}
\crefname{construction}{Construction}{Constructions}
\crefname{subsection}{Subsection}{Subsections}
\crefname{section}{Section}{Sections}
\crefname{claim}{Claim}{Claims}
\crefname{clm}{Claim}{Claims}
\crefname{algorithm}{Algorithm}{Algorithm}
\crefname{definition}{Definition}{Definitions}

\newcommand{\decprob}[3]{
   \begin{center}%
     \begin{itemize}[d]
       \item[\textsc{#1}]
       \item[\textbf{Input:}]  #2\\[0.2ex]
       \item[\textbf{Question:}]  #3
     \end{itemize}
  \end{center}
}

\newcommand{\scov}{\textsc{Set Cover}\xspace}

\newcommand{\manyone}{\textsc{Many-To-One Matching}\xspace}
\newcommand{\manyones}{\textsc{MM}\xspace}
\newcommand{\sumcol}[1]{{\color{red!50!black}#1}}
\newcommand{\maxcol}[1]{{\color{blue!50!black}#1}}
\newcommand{\minsummatching}[1][$\Pi$]{\textsc{Min{\sumcol{Sum}} Cap #1}\xspace}
\newcommand{\minmaxmatching}[1][$\Pi$]{\textsc{Min{\maxcol{Max}} Cap #1}\xspace}

\newcommand{\popexist}{\textsc{PM-cap}}

\newcommand{\sumpo}{\textsc{{\minsum}Par-p}}
\newcommand{\maxpo}{\textsc{{\minmax}Par-p}}
\newcommand{\sumpop}{\textsc{{\minsum}pop-p}}
\newcommand{\maxpop}{\textsc{{\minmax}pop-p}}

\newcommand{\minsum}{\textsc{Min\-\sumcol{Sum}}}
\newcommand{\minmax}{\textsc{Min\-\maxcol{Max}}}
\newcommand{\avgrank}{\textsc{AvgRank}}
\newcommand{\maxsm}{\textsc{CarSM}}
\newcommand{\lone}[1][$\pr$]{\ensuremath{|#1|_{\sumcol{1}}}}
\newcommand{\lmax}[1][$\pr$]{\ensuremath{|#1|_{\maxcol{\infty}}}}
\newcommand{\tdm}{\textsc{3dm}}

\newcommand{\vect}[1]{\ensuremath{\boldsymbol{#1}}}

\newcommand{\px}{\ensuremath{\vect{x}}}
\newcommand{\py}{\ensuremath{\vect{y}}}

\newcommand{\pr}{\ensuremath{\vect{r}}}

\newcommand{\quota}{\ensuremath{\boldsymbol{q}}}
\newcommand{\acset}{\ensuremath{\mathsf{A}}}

\newcommand{\sumcap}{\ensuremath{k^{+}}}
\newcommand{\maxcap}{\ensuremath{k^{\max}}}

\newcommand{\seqq}[1]{\ensuremath{\langle #1 \rangle}}

\newcommand{\myemph}[1]{{\color{purple!30!black}\emph{#1}}}

\definecolor{darkgreen}{rgb}{0.01,0.6,0.1}
\definecolor{darkblue}{rgb}{0,0,0.4}
\definecolor{winered}{rgb}{0.6,0.1,0.1}
\definecolor{doncolor}{RGB}{78,154,0}
\definecolor{falsecolor}{RGB}{0,55,255}
\definecolor{truecolor}{RGB}{164,0,0}
\definecolor{lightblue}{rgb}{0.527,0.805,0.977}

\tikzset{ttrue/.style={color=truecolor!50!white}}
\tikzset{ffalse/.style={color=falsecolor!50!white}}
\tikzset{dont/.style={color=doncolor!50!white}}
\tikzset{trueline/.style =   {line width= 3pt, ttrue}}
\tikzset{falseline/.style =   {line width= 3pt, ffalse}}
\tikzset{dontline/.style =   {line width= 3pt, dont}}

\begin{document}
\maketitle
\begin{abstract}
  We consider many-to-one matching problems, where one side corresponds to applicants who have preferences and the other side to houses who do not have preferences. 
  We consider two different types of this market: one, where the applicants have capacities, and one where the houses do.
 First, we answer an open question by  \cite{manlove2006popular} (partly solved \cite{paluch2014popular} for preferences with ties), that is, we show that deciding if a popular matching exists in the house allocation problem, where agents have capacities is NP-hard for previously studied versions of popularity.
  
  Then, we consider the other version, where the houses have capacities. We study how to optimally increase the capacities of the houses to obtain
  a matching satisfying multiple optimality criteria, like popularity, Pareto-optimality and perfectness.
  
  We consider two common optimality criteria, one aiming to minimize the sum of capacity increases of all houses (abbrv.\ as \minsum) and the other aiming to minimize the maximum capacity increase of any school (abbrv.\ as \minmax).
  We obtain a complete picture in terms of computational complexity and some algorithms.
  
\end{abstract}




\pagestyle{plain}


\section{Introduction}\label{sec:intro}

Two-sided matching markets with one-sided preferences arise in many applications such as house allocation, school choice and more. 
Such problems often correspond to allocating a set $H$ of heterogeneous indivisible goods
among a set $A$ of agents or to assigning a set $A$ of applicants to a set $H$ of houses/jobs. The agents/applicants are assumed to have some strict preferences over the goods/houses they find acceptable. This is known as the house allocation (\textsc{ha}) problem.
The house allocation problem is a widely studied problem in the literature especially nowadays, as it has been started to be used to model kidney exchange programs too. 

Such frameworks are also used in allocating students to schools, graduates to trainee positions, professors to offices, clients to servers,
etc.
In many of these applications, one or both sides of the market are allowed to have capacities. To unify terminology and follow the literature, we call one side of the market the applicants, and the other side the houses.

Our goal is to find a \emph{good} assignment between $A$ and $H$ without violating the capacity constraints.


As to what defines a good matching, the answer varies from application to application.
The simplest concept is to ensure that every applicant is matched (or saturated if they have capacities), we call such allocations \myemph{perfect}.
Note that having a perfect allocation is particularly important in school choice since every student should be admitted to some school. 
Another widely used concept is \myemph{Pareto-optimality}. A Pareto-optimal allocation is such that there is no other allocation, where all of the applicants are better off.
Of course, Pareto-optimality is very desirable concept, that most mechanisms aim to satisfy. 

Finally, an other well known and studied concept is
\myemph{popularity}. We say that a matching $M$ is \myemph{popular}, if there is no other matching $M'$, such that $M$ would lose in a head to head comparison with $M'$, that is, more agents would prefer $M'$ to $M$ than the other way around.
If the applicants have unit capacities, then comparing two allocations is straightforward, but otherwise the concept of popularity can vary by definition. In this  paper, we consider both the traditional version of the definition, introduced by \cite{brandl2016popular} and another natural one studied by \cite{paluch2014popular}, which compares two allocations in a lexicographic way.

It is not hard to see, that if an assignment $M$ is popular, then it also must be Pareto-optimal. Clearly, a matching $M'$, where each applicant (weakly) improves and at least one strictly improves dominates $M$ with respect to any natural definition. However, a Pareto-optimal or a popular matching may need not be perfect.

\smallskip
\noindent \textbf{Our contributions.}

If the applicants can have capacities, then we show that even finding a popular matching is NP-hard, answering an open question of \cite{manlove2006popular}. This holds even if the capacities of the applicants are at most 3 in the traditional version and at most 2 in the lexicographic one. 

Therefore, we mainly focus on the version, where only the houses are allowed to have capacities.

We study the question of how to optimally modify the capacities of the houses, such that the resulting instance admits a perfect and Pareto optimal, or a perfect and popular matching.  We consider two optimality criteria: one aiming to minimize the sum of the capacity changes and one aiming to minimize the maximum of the capacity changes. The case of Pareto-optimality is rather easy, as there are always maximum size Pareto-optimal matchings, so the trivial lower bound, which is minimum the number of applicants who are not matched in an allocation is tight. For popularity we obtain some nice and surprising results. We show that if we are only allowed to increase the capacities, then minimizing the sum of changes is polynomial-time solvable. However, if we are allowed to decrease the capacities too (which in some cases can be beneficial as we will illustrate on some examples), then the problem becomes NP-hard. Surprisingly, if the maximum of the capacity changes is minimized, then both versions become NP-hard and also hard to appriximate within a constant factor. This is in sharp contrast with recent results for two-sided preferences and stable and (one-sided)-popular matchings, where the \minmax\ case is solvable and the \minsum\ case is NP-hard \cite{chen2023optimal}.





\smallskip
\noindent \textbf{Related work.}

The popular matching problem was introduced first for markets with unit capacities and two-sided preferences by \cite{gardenfors1975popular}. Popularity has been extended to many-to-many instances later by \cite{brandl2016popular}. 

Popularity has also been explored in the house allocation model, first with unit capacities \cite{abraham2007popular}, then with arbitrary capacities on the houses \cite{manlove2006popular}, where they also posed as an open question whether a popular matching can be find in polynomial-time if the side with preferences has non-unit capacities.
This question has been partly solved by \cite{paluch2014popular}, but only for preferences which may contain ties. She also defined an extension of popularity to many-to-one instances, where agents compare matchings in a lexicographic fashion.

Achieving optimal allocation by capacity changes has been studied a lot in recent literature, although mostly for markets, where both sides have preferences.
Closest to our work is \cite{chen2023optimal}, where they also studied how to optimality increase the capacities to obtain matchings that satisfy multiple optimality criteria simultaneously. They studied markets with two-sided preferences, where one optimality notion was always stability. Our work investigates similar problems with markets, where only one side has preferences. 

There are further papers examining capacity variations.
\cite{LN2022} studied student-school allocations, where each seat at the schools have some cost, and the aim is to assign all students while minimizing the cost.
\cite{RLPC2014} propose a seat-extension mechanism to increase student's welfare.
\cite{KHHSUY2017quota} design a strategy-proof mechanism to address minimum and maximum quotas.
\cite{NV2018} study many-to-one matching with couples and propose algorithms to find a stable matching by perturbing the capacities.
\cite{BCLT2022capvariation,BCLRT2022capplanning} consider capacity variations to obtain a stable matching with minimum sum of the ranks of the matched schools (\avgrank) or maximum cardinality (\maxsm). 
\cite{AKI2022CapExp} propose some alternative method and conduct experiments for \minsum \avgrank. \cite{yahiro2020game} considered the problem where there are a fixed number of resources and the capacities of the schools are based on how much resource we allocate them. 

\smallskip
\noindent \textbf{Paper structure.}
We start by introducing the main concepts and definitions in \cref{sec:defi}. Then, in \cref{sec:appl-capac} we prove NP-hardness of deciding if a popular matching exists if agents have capacities. In \cref{sec:house-capac} we consider the case when only the houses have capacities. We investigate problems, where our goal is to guarantee the existance of desired solutions by perturbing the capacities as little as possible. In the case when we aim for a perfect and popular matching, we provide an algorithm, if the objective is to minimize the sum of changes and show NP-hardness if we want to minimize the maximum of the changes. To have Pareto-optimality and perferctness, we give a simple solution to both problems.
\section{Basic definitions and fundamentals} 
\label{sec:defi}

By $\mathds{N}$ we mean the set of all positive integers. 
Given an integer~$t$, let \myemph{$[t]$} denote the set~$\{1,\cdots,t\}$.
Given two integer vectors~$\px, \py$ of dimension~$t$, i.e., $\px,\py \in \mathds{Z}^{t}$,
we let \myemph{$\px + \py$} denote the new integer vector~$\pr$ with $\pr [z] = \px[z]+\py[z]$ for all $z\in [t]$,
and we write \myemph{$\px\le \py$} if for each index~$i\in [z]$ it holds that
$\px[i]\le \py[i]$; otherwise, we write \myemph{$\px \not\le \py$.}

If $S=\{s_1,\ldots,s_{|S|}\}$ is a set of indexed elements, then $\seqq{S}$ always denotes the increasing order~$s_1,\ldots,s_{|S|}$.

In our hardness reductions, we will use the NP-hardness of the following problem.

\decprob{\tdm }
{%
  Elements $A=\{ a_1,\dots,a_{\hat{n}}\}$, $B=\{ b_1,\dots,b_{\hat{n}}\}$, $C=\{ c_1,\dots,c_{\hat{n}}\}$ and a family $\mathcal{S}=\{ S_1,\dots,S_{3\hat{n}}\}$ of 3-element subsets of $A\cup B\cup C$ such that $|S\cap A|=|S\cap B|=|S\cap C|=1$ for each $S\in \mathcal{S}$ and for each element in $A\cup B\cup C$, there are exactly 3 sets $S_j\in \mathcal{S}$ covering it.
}
{%
 Is there a subset $\mathcal{S}'\subset \mathcal{S}$ (called an exact 3-cover), such that each element in $A\cup B\cup C$ is contained in exactly one set $S\in \mathcal{S}'$?
}
The NP-hardness of this special variant is proved in \cite{gonzalez1985clustering}. They showed NP-hardness remains for \textsc{x3c} (which is the more general version of \tdm, where the elements are not partitioned into three classes), even if each element is in exactly three sets. Their reduction also works in the case when the elements have to be partitioned into three classes, so the hardness of this \tdm\ variant also follows.

A well know NP-hard optimization problem is \scov.

\decprob{\scov }
{%
  Elements $E=\{ e_1,\dots,e_{\hat{n}}\}$ and a family $\mathcal{S}=\{ S_1,\dots,S_{\hat{m}} \}$ of subsets of $E$ and a number $k$.
}
{%
 Is there a subset $\mathcal{S}'\subset \mathcal{S}$ with $|\mathcal{S}'|\le k$, such that each element in $E$ is covered by at least one $S_j\in \mathcal{S}'$?
}

It is also known (\cite{feige1996threshold},\cite{dinur2014analytical}) that \scov\ is NP-hard to approximate within any constant factor $d$.

\subsection{Many-to-one matching}
The \manyone\ (in short, \manyones) problem has as input a set $A=\{a_1,\ldots, a_{n}\}$ of $n$ applicants and a set~$H=\{h_1,\ldots,h_{m}\}$ of $m$ houses. We assume that the acceptability relations are given by a bipartite graph $G=(A,H,E)$, where $(a,p)\in E$ if and only if $a$ finds house $p$ acceptable.

We destinguish two \manyone\ models, one where the applicants have capacities and one where the houses. 
In both variations, we are given strict preference orders $\succ_a$ for each applicant $a\in A$ over her acceptable houses, which we denote by $\acset (a)$. We also use $\succ_a$ to denote the induced strict order of $a$ over her incident edges. For an applicant or house $x\in A\cup H$, we denote by $E(x)$ the edges incident to $x$ in $G$. We are given a vector $\quota^A \in \mathbb{N}^n$, where $\quota^A [i]$ is the capacity of applicant $a_i$ and a vector $\quota^H \in \mathbb{N}^m$, where $\quota^H [i]$ denotes the capacity of house $h_i$.
In the first case, we assume $\quota^H[i]= 1$ $\forall i\in [m]$, while in the second case, $\quota^A[j]=1$, $\forall j\in [n]$. In the case of $\quota^A\equiv 1$, we omit $\quota^A$ from the input and write $\quota := \quota^H$. If $\quota^H\equiv 1$, then we omit $\quota^H$ from the input and write $\quota := \quota^A$

An \myemph{ (feasible) assignment (or simply matching)}~$M$ is a subset of the edges, such that $|M\cap E (a_i)|\le \quota^A[i]$ for each $i\in [n]$ and $|M\cap E (h_j)|\le \quota^H[j]$ for each $j\in [m]$. We use the notation $M(x)$ to refer to $M\cap E(x)$, i.e. the incident edges of $M$ to $x$.

We say that an applicant $a_i$ or a house $h_j$ is \emph{saturated} in $M$, if $|M\cap E (a_i)|= \quota^A[i]$ or $|M\cap E (h_j)|= \quota^H[j]$ holds respectively, otherwise she/it is \emph{unsaturated}.
We assume that each applicant~$a_i$ prefers to be assigned some acceptable houses rather than not being assigned at all.

Given a  \manyones\ instance~$I=(A,H, (\succ_a)_{a\in A}, \quota^A,\quota^H)$ and a matching~$M$ for $I$,
we say $M$ is \myemph{perfect} if every applicant is saturated in~$M$ (in the case when the applicants have capacity one, this just corresponds to them being assigned a house).

\smallskip
\noindent \textbf{Pareto optimality and Popularity}

Suppose each applicant has capacity one.
We say a matching~$M$ \myemph{Pareto-dominates} another matching~$M'$ if the following holds:
\begin{compactitem}[--]
  \item for each applicant~$a$ it holds that either $M(a)=M'(a)$ or $M(a)\succ_a M'(a)$, and
  \item at least one applicant~$a$ has $M(a)\succ_a M'(a)$.
\end{compactitem}
We call a matching~$M$ \myemph{Pareto-optimal} if it is not Pareto-dominated by other matchings. To define Pareto-optimality for applicants with capacities, we have to extend their preferences to subsets. We do not define this here, as we only cosider Pareto-optimality with unit capacity applicants.

Next we define popularity. Let us start with the case, when applicants have capacity one. In this model, we compare two matchings $M$ and $M'$ in the following way. Each applicant $a\in A$ cast a vote $vote_a(M,M')\in \{ -1,0,1\}$, such that $vote_a(M,M')=0$, if $M(a)=M'(a)$, $vote_a(M,M')=1$, if $M(a)\succ_aM'(a)$ and $vote_a(M,M')=-1$ otherwise.

If applicants have capacities, then we can allow applicant $a_i$ to cast a vote from $\{ -\quota^A [i],-\quota^A [i]+1,\dots,-1,0,1,\dots,\quota^A [i]\}$ to quatintify her improvement.
We define two different models. The first one is the one used by \cite{brandl2016popular}, who first extended the notion of popularity to many-to-many instances. This notion has also been used since by many works such as \cite{kamiyama2020popular}, \cite{csaji2022solving}. This is defined as follows. Let $S,T\subseteq E (a)$. We define a \emph{feasible pairing} $N$ as a matching (or pairing) between the elements of $S\setminus T$ and $T\setminus S$ that satisfies that $|N|=\min \{ |S\setminus T|,|T\setminus S| \}$ and that each element in $S\setminus T$ and $T\setminus S$ is only paired with at most one element from the other set. For a feasible pairing $N$, we define $vote_a(S,T,N)=|\{ xy\in N \mid x\in S\setminus T, y\in T\setminus S, x\succ_a y\}|-|\{ xy\in N \mid x\in S\setminus T, y\in T\setminus S, y\succ_a x\}|+|S\setminus T|-|T\setminus S|$. Then, $vote_a(M,M')$ is defined by taking the worst possible pairing $N$ for $M(a)$ and $M' (a)$ (from the point of view of $M$), that is $vote_a(M,M')=\min_{N} \{vote_a(M  (a),M' (a),N) \}$. 

The second notion is based on comparing matchings lexicographically. This has been also used in previous works such as \cite{paluch2014popular} and \cite{csaji2022solving}. Here, we first extend each applicants preferences to a strict preference order over the possible sets of incident edges (i.e. $2^{E(a)}$). Let $S,T\subseteq E(a)$. We say that $a$ \emph{lexicographically prefers} $S$ to $T$, denoted as $S\succ^{lex}_aT$, if the best element according to the order $\succ_a$ of the symmetric difference $S\triangle T=(S\setminus T)\cup (T\setminus S)$ is in $S$. Then, we compare two matchings $M$ and $M'$ in the following way. Each applicant $a\in A$ cast a vote $vote_a(M,M')\in \{ -1,0,1\}$, such that $vote_a(M,M')=0$, if $M(a)=M'(a)$, $vote_a(M,M')=1$, if $M(a)\succ^{lex}_aM' (a)$ and $vote_a(M,M')=-1$ otherwise.

Finally, we define $vote(M,M')$ to be the sum of votes of all applicants in both models, that is $\sum_{a\in A}vote_a(M,M')$. We call a macthing $M$ \emph{popular} (in both cases), if there is no other matching $M'$, such that $vote(M,M')<0$ holds. In that case, we say that $M'$ \emph{dominates} $M$. To distinguish the two different notions of popularity, we call the second one \emph{lexicographic popularity}.

\subsection{Studied problems}

First, we study the existence of popular matchings with capacitated applicants. 
\decprob{\popexist}
{%
  An instance~$I=(A,H,(\succ_a)_{a\in A},\quota =\quota^A,\quota^H\equiv 1 )$ of \manyones\ with capacitated applicants.%
}
{%
  Is there a a popular matching $M$? 
}

Of course, the complexity of \popexist\ may depend on which notion of popularity we use (popularity or lexicographic popularity). However, we show that for both traditional notions of popularity, \popexist\ is NP-hard. This is in sharp contrast with the result of \cite{brandl2016popular}, where they showed that there is always a popular matching with their definition of popularity and it is always possible to find one in polynomial time. The difference between our framework and theirs is that now only one side of the market has preferences and therefore only one side votes, whereas in their model, both sides have preferences and vote. 

Next, we consider the following decision problems, where $\Pi \in \{ $Pareto-optimal and Perfect, Popular and Perfect
$\}$.

\decprob{\minsummatching ({\normalfont\text{resp. }}\minmaxmatching)}
{%
  An instance~$I=(A,H,(\succ_a)_{a\in A},\quota^A\equiv 1,\quota = \quota^H )=(A,H,(\succ_a)_{a\in A},\quota )$ of \manyones\ and a capacity bound~$\sumcap\in \mathds{N}$ (resp. $\maxcap\in \mathds{N}$).%
}
{%
  Is there a capacity change vector~$\pr$ with $\lone\le \sumcap$ (resp.\ $\lmax \le \maxcap$) s.t.\ $I'=(A,H,(\succ_a)_{a\in A}, \quota+\pr)$ admits a $\Pi$~matching. 
}

We abbreviate the problems \minsummatching[Pareto-optimal and perfect], \minmaxmatching[Pareto-optimal and perfect ], \minsummatching[Popular and perfect], \minmaxmatching[Popular and perfect]
with \sumpo,  \maxpo, \sumpop\ and \maxpop\ respectively.

We consider two different types of \sumpop\ and \maxpop\ with respect to the allowed capacity change vectors. In the first case, we require that $\pr \ge 0$, that is each capacity can only be increased, while in the second one, capacities can be both increased and decreased. We call a capacity change vector witnessing a YES instance a \emph{good capacity change vector}. We further refer to it as a \emph{good capacity increase/decrease vector}, if $\pr \ge 0$ or $\pr \le 0$ holds respectively.

\section{Applicants have capacities}
\label{sec:appl-capac}
In this section we start by investigating the popular matching problem, where the applicants are allowed to have capacities. We show that the problem of deciding if a popular matching exists becomes NP-hard, even in very restricted settings and small constant capacities.

We suppose each house has unit capacity in this model (all our hardness results imply hardness for the case, where house capacities are also allowed). As only the applicants have capacities now, for simplicity we abbriviate $\quota^A[i]$ as $\quota [a_i]$. 

\subsection{Traditional popularity}

We start by showing that verification remains polynomial-time solvable, hence the problem is in NP. The proof follows similar arguments as the proof in \cite{kiraly2017finding}, where they show that in the many-to-many model, popularity can be verified in polynomial-time.
\begin{theorem}
Given a matching $M$, we can verify in polynomial-time whether $M$ is popular.
\end{theorem}
\begin{proof}
    We create a graph $G^M=(A',H,E')$ from the acceptability graph $G$ and the matching $M$ as follows. For each node $a_i\in A$, we make $q_i=\quota [a_i]$ copies of $a_i$, called $a_i^1,\dots,a_i^{q_i}$. The nodes of $H$ remain the same. In this new graph, every capacity is 1. Then, for each edge $(a_i,p_j)\in E\setminus M$, we add edges $(a_i^1,p_j),\dots,(a_i^{q_i},p_j)$ to $E'$.
    
    We create a matching $\hat{M}$ in $G^M$ by adding further edges to $E'$ as follows.
    For each edge $(a_i,p_j)\in M$, we add an $(a_i^l,p_j)$ edge to $\hat{M}$, such that these edges form a matching in $G^M$. Define $M(a_i^l)$ to be the edge of $M$, whose copy is adjacent to $a_i^l$ in $\hat{M}$, if there is such an edge, otherwise define it to be $\emptyset$.

    Now, we define a weight function $w$ over $E'$. Let $(a_i^l,p_j)\in E'$. Then, $w(a_i^l,p_j)=vote_{a_i}(M(a_i^l),(a_i,p_j))$. 

    Let us call a cycle $C$ an \emph{alternating cycle}, if the edges of $C$ alternate between $M'$ and $E'\setminus \hat{M}$. We define \emph{alternating paths} analogously. Finally, for an alternating path $P$, let $mod(P)$ be the number of endpoints in $A'$ of $P$ covered by $P\cap \hat{M}$ minus the number of endpoints in $A'$ covered by $P\setminus \hat{M}$. That is, $mod(P)$ can be $-2,-1,0,1$ or $2$, depending on how many of the endpoints of $P$ are covered by $P\cap \hat{M}$ and $P\setminus \hat{M}$.

    \begin{claim}
    \label{claim:popverif}
      $M$ is popular, if and only if there is no alternating cycle $C$ such that $w(C)<0$ and there is no alternating path $P$ such that $w(P)+mod(P)<0.$
        
    \end{claim} 
    \begin{proof}
    Suppose first that $M$ is not popular and there is a matching $ \mu$ that dominates it. Fix a feasible pairing $N_i$ for each $a_i\in A$ between the adjacent edges of $M\setminus \mu$ and $\mu \setminus M$. Then, create a matching $\hat{\mu}$ in $G^M$ by adding the copy $(a_i^l,p_j)$ to $\hat{\mu}$ for each edge $(a_i,p_j)\in \mu$ such that $M(a_i^l)$ is paired with $(a_i,p_j)$ in $N_i$, if $(a_i,p_j)$ is paired with an edge in $N_i$, otherwise we choose a copy $(a_i^l,p_j)$ such that $a_i^l$ is not covered by $\hat{M}$ nor by a previously added edge of $\hat{\mu}$. (We can do this, as there are enough copies of the vertices). 

    Let $\hat{M}\triangle \hat{\mu}$ consist of alternating cycles $C_1,\dots,C_p$ and alternating paths $P_1,\dots,P_q$.
    Now, it is easy to see that $vote(M,\mu)<0$ is exactly equal to $\sum_{i=1}^pw(C_i)+\sum_{j=1}^q(w(P_j)+mod(P_j))$. Hence, we get there there is an alernating cycle or path satisfying the condition.
  
    For the other direction, suppose first that there is an alternating cycle $C$ with $w(C)<0$. As $C$ is an alternating cycle, it cannot happen that there is an edge $(a_i,p_j)\in E$, such that two copies of it are in $C$. This holds because only the edges in $E\setminus M$ have multiple copies, and only one of the adjacent edges of $p_j$ in $C$ is a copy of an edge in $E\setminus M$. Therefore, as no edge $(a_i,p_j)\in E$ has more than one copies in $C$, we can construct a matching $M'$ by exchanging the edges of $C\cap M$ and $C\setminus M$ with feasible pairings defined by the two adjacent edges for each copy of each applicant, such that the sum of votes is negative, so $M$ is not popular.

    If there is an alternating path $P$ with $w(P)+mod(P)<0$, then again, it is easy to see that each edge has at most one copy in $P$, thus we can create again a matching $M'$ that dominates $M$.

    \end{proof}
Using \cref{claim:popverif}, it is enough to decide whether such an alternating path or cycle exists. For this, first orient the edges of $G^M$ such that the edges in $\hat{M}$ point towards $H$, the other edges point towards $A'$.

Then, the existence of such an alternating cycle can be decided with the Bellman-Ford algorithm, applied to this edge-weighted directed graph $G^M$, as every directed cycle must be alternating. 

Suppose that no negative weight cycle exist. This means that $w$ is conservative, so the Bellman Ford algorithm can find shortest paths between any two nodes. Then, to decide whether an alternating path $P$ exist with $w(P)+mod(P)<0$, we can find shortest paths between all pairs of nodes, compute $mod(P)$ for them and check whether $w(P)+mod(P)<0$ holds.
    
\end{proof}

We procede by showing NP-hardness of the existence question.

\begin{theorem}
\label{thm:onesided-pop}
\popexist\ is NP-complete, even if each applicant has capacity 3, each applicant has three acceptable houses, and each house is acceptable to three applicants.
\end{theorem}
\begin{proof}
We reduce from \tdm. Let $I$ be an instance of \tdm. 
Make an instance $I'$ of \popexist\ as follows.

\begin{itemize}
    \item [--] For each set $S_j\in \mathcal{S}$, we add an applicant $s_j$ with capacity 3,
    \item [--] For each element $a_i\in A$, we add a house $a_i$, for each element $b_i\in B$, we add a house $b_i$ and for each element $c_i\in C$, we add a house $c_i$.
\end{itemize}
The acceptability relations correspond to the incidence relations between the 3-sets and the elements.
The preferences of the $s_j$ applicants are defined in a way such that they rank their $c$-type house first, their $b$-type house second and their $a$-type house third.

We claim that there is an exact cover $\mathcal{S}'$ in $I$ if and only if $I'$ admits a popular matching $M$.

First suppose that there is an exact 3-cover $\mathcal{S}'\subset \mathcal{S}$ in $I$ consisting of 3-sets $S_{j_1},\dots,S_{j_{\hat{n}}}$. Define a matching $M$ of $I'$ by giving the applicants $s_{j_1},\dots,s_{j_{\hat{n}}}$ all their acceptable houses. The other applicants receive no houses. 

We claim $M$ is popular. Indeed, the applicants corresponding to the 3-sets of $\mathcal{S}'$ receive all their acceptable houses, so their vote for a matching $M'$ is $-k$ if they lose $k$ of their houses in $M'$. For all other applicants, they receive no house initially, so their vote is $k$, if they get $k$ houses in $M'$. As all houses were assigned originally, it follows that the sum of votes can only be at most 0 for any matching $M'$ of $I'$.

For the other direction suppose there is a popular matching $M$. First, observe that each house must be allocated, as $M$ must be maximal (if a house is not allocated, then its adjacent applicants cannot be saturated). As each $s_j$ applicant has only one acceptable $a$-type house, it follows that there is exactly $\hat{n}$ applicants $s_{j_1},\dots,s_{j_{\hat{n}}}$ who obtain their worst houses. Suppose that there is one of them, say $s_{j_1}$, who does not obtain all her acceptable houses. Then, make a matching $M'$ from $M$, by giving $s_{j_1}$ a house she does not have in $M$ (by taking it from someone else), and giving her $a$-type house away to an applicant who considers it acceptable (this applicant was unsaturated in $M$). Then, two applicants vote with $+1$ for $M'$ and at most one applicant (who looses the house we give $s_{j_1}$) votes with $-1$. This contradicts the fact that $M$ is popular. Hence, all of $s_{j_1},\dots,s_{j_{\hat{n}}}$ obtain all their houses, which is only possible if $S_{j_1},\dots,S_{j_{\hat{n}}}$ form an exact 3-cover in $I$. 
\end{proof}

\subsection{Lexicographic popularity}

Next, we move on to lexicograpic preferences.
It was already shown by \cite{paluch2014popular} that verifying whether a matching is lexicographic popular can be done in polynomial-time, hence the problem is in NP. She also proved hardness, but only if ties are allowed in the preference list. Here, we strengthen her result by showing that the problem remains NP-hard even with strict preferences.

\begin{theorem}
\label{thm:onesided-lex}
\popexist\ for lexicographic popularity is NP-hard, even if each applicant has capacity 2, each applicant has at most 3 acceptable houses and each house is acceptable to at most 3 applicants.
\end{theorem}
\begin{proof}
We reduce from \tdm.  Let $I$ be an instance of \tdm.. 
Make an instance $I'$ as follows.

\begin{itemize}
    \item [--] For each set $S_j\in \mathcal{S}$, we add three applicants $s_j^1,s_j^2,s_j^3$ with capacity 2, along with three houses $h_j^1,h_j^2,h_j^3$.
    \item [--] For each element $a_i\in A$, we add a house $a_i$, for each element $b_i\in B$, we add a house $b_i$ and for each element $c_i\in C$, we add a house $c_i$.
\end{itemize}
The acceptability relations are the following. For each set $S_j=\{ a_i,b_k,c_l\}$, we have the edges $(s_j^1,a_i),(s_j^2,b_k),(s_j^3,c_l)$ in the acceptability graph $G^{acc}$. Furthermore, for each $j$ and each $\ell \in [3]$ we also have the edges $(s_j^{\ell},h_j^{\ell}),(s_j^{\ell},h_j^{\ell +1})$ in $G^{acc}$, where $\ell +1$ is taken modulo 3.
The preferences of the $s_j^{\ell}$ applicants are defined in a way such that they rank their $h_j^{\ell}$ house first, their corresponding $a_i$/$b_k$/$c_l$ house second and their $h_j^{\ell +1}$ house third. 

We claim that there is an exact cover in $I$ if and only if $I'$ admits a popular matching $M$.

First suppose that there is an exact 3-cover $\mathcal{S}'$ in $I$ consisting of 3-sets $S_{j_1},\dots,S_{j_{\hat{n}}}$. Define a matching $M$ of $I'$ by giving the applicants $s_{j_1}^{\ell},\dots,s_{j_{\hat{n}}}^{\ell}$, $\ell \in [3]$ their best two houses (which are distinct, as $\mathcal{S}'$ is an exact 3-cover and first houses are distinct). The other $s_j^{\ell}$ applicants only receive their best house $h_j^{\ell}$.  

We claim $M$ is popular. As the applicants $s_{j_1}^{\ell},\dots,s_{j_{\hat{n}}}^{\ell}$, $\ell \in [3]$ are saturated with their best two houses, if they receive a different set of houses in a matching $M'$, they vote with $-1$. The other applicants also vote with $-1$, whenever they loose their best (and only) $h_j^{\ell}$ house in $M'$. Also, only these applicants can improve. Let $M'$ be any matching and let $p$ denote the number of improving applicants. Then, all these $p$ applicants must receive an additional house in $M'$, while keeping their best one. As these applicants cannot take any house from the set $h_{j_1}^{\ell},\dots,h_{j_{\hat{n}}}^{\ell}$, $\ell \in [3]$ (they do not find it acceptable), it follows that at least $p$ applicants loose a house (any applicant with two houses have a house from  $\{ h_{j_1}^{\ell},\dots,h_{j_{\hat{n}}}^{\ell}$, $\ell \in [3]\}$. By our observations, this implies that the sum of votes cannot exceed 0. Hence, $M$ is popular.

For the other direction suppose there is a popular matching $M$. 
\begin{claim}
\label{claim:onesided-lex}
For any popular matching $M$,
\begin{enumerate}
    \item each $h_j^{\ell}$ house must be allocated to $s_j^{\ell}$,
    \item each house must be allocated,
    \item if applicant $s_j^{\ell}$ obtains two houses, than so does $s_j^{\ell -1}$ (addition taken modulo $3$)
\end{enumerate}
\end{claim}
\begin{proof}
Suppose $s_j^{\ell}$ does not get $h_j^{\ell}$. Then, $h_j^{\ell}$ is allocated to $s_j^{\ell -1}$, otherwise $s_j^{\ell}$ could switch to it without making anyone worse off, contradicting popularity. Hence, $s_j^{\ell -1}$ can only have one other house, so she does not obtain one of his best two houses. Therefore, if $s_j^{\ell -1}$ gives house $h_j^{\ell}$ to $s_j^{\ell}$, $s_j^{\ell}$ drops her houses in $M$, and $s_j^{\ell -1}$ takes a house she did not have in $M$, which is better than $h_j^{\ell}$ (maybe from some other applicant), then we obtain a matching $M'$ such that two applicants prefer $M'$ to $M$ and at most one prefer $M$ to $M'$, contradiction. 

Next, assume that there is a house that remains free. By (1), it can only be an $a/b/c$-type house. However, in $M$, no such houses can remain free, because otherwise any applicant that considers it acceptable would be unsaturated, hence she could improve without making anyone else worse off.

Finally, suppose that $s_j^{\ell}$ obtains her corresponding $a/b/c$-type house, but $s_j^{\ell -1}$ does not. Then, $s_j^{\ell -1}$ is unsaturated. Let $M'$ be the matching we obtain by giving $h_j^{\ell}$ to $s_j^{\ell -1}$ and the other house of $s_j^{\ell}$ to any applicant who considers it acceptable. Then, two applicants improve and only one applicant gets worse in $M'$, contradicting that $M$ is popular.
\end{proof}

By \cref{claim:onesided-lex}, each house that corresponds to elements is assigned and for each set $S_j$ either all 3 of its corresponding element houses are matched to the applicants corresponding to $S_j$ or none of them are. Hence, $M$ defines a an exact 3-cover in $I$.
\end{proof}



\section{Capacitated house allocation }
\label{sec:house-capac}
In this section we look at the optimization problems related to capacity modifications.

\subsection{Pareto optimal and Perfect matchings}

First, we investigate the complexity of \sumpo\ and \maxpo. It turns out that these problems are rather easy.

\begin{theorem}
\sumpo\ and \maxpo\ can be solved in polynomial time, either if decreases are allowed or not.
\end{theorem}
\begin{proof}
We start with \sumpo.
It is widely known that there is always a Pareto optimal matching of maximum size (weigh each edge according to the preferences, then a maximum weight maximum size matching is Pareto-optimal). Hence, decreases are never beneficial and the problem reduces to find a minimal sum capacity increase vector $\pr$, such that $\quota +\pr$ admits an applicant-perfect matching. 
Let $x$ be the size of a maximum size matching and let $n$ be the number of applicants. Clearly, the optimum value is at least $n-x$. It is easy to find a capacity increase vector, with $|\pr |_1=n-x$ by just assigning each applicant that is left out in a maximum size matching $M$ to a house they find acceptable and modify the capacities accordingly. As there is always a maximum size Pareto optimal matching, this new instance will admit an applicant perfect and Pareto optimal matching too. 

For \maxpo\, we can iterate through $\maxcap =1,2,\dots, n$ and find the smallest one, where the instance admits an applicant perfect matching in polynomial time. Then, we can find one that is also Pareto-optimal.
\end{proof}

\subsection{Popular and perfect matchings}

\subsection{Structural results}
Now we move on to the more interesting problems, \sumpop\ and \maxpop. We start by showing, that if only increases are allowed, than \sumpop\ can be solved in polynomial-time. Let us denote the capacity $\quota^H[j]$ of $h_j$ by $\quota [h_j]$.

\cite{manlove2006popular} provided a nice characterization for popular matchings, that we are going to utilize now. We formulate it in a slightly different way, to better suit our needs. Let $G=(A,H,(\succ_a)_{a\in A},\quota )$ be an instance of the capacitated house allocation problem. For each applicant $a\in A$ define $f(a)$ to be the house that is most preferred by $a$. For a house $h$, we call an applicant $a$ an \emph{admirer of $h$}, if $h=f(a)$.
Next, define $s(a)$ to be either $f(a)$, whenever $h=f(a)$ has at most as many admirers as its capacity $\quota [h]$, or the most preferred house by $a$ among those that have strictly less admirers then their capacity, if there is any such acceptable house, otherwise let $s(a)=\emptyset$ (we do not create a last resort house for such applicants as Manlove and Sng \cite{manlove2006popular}, and so we do not require all applicants to be matched).  

Make a graph $G'$ by adding the edges $(a,f(a))$ and $(a,s(a))$ for each applicant $a\in A$. In the case when $f(a)=s(a)$ this corresponds to adding two parallel edges. We add both these edges to make the characterization more simple to formulate.
\begin{theorem}[\cite{manlove2006popular} reformulated]
\label{thm:popperf}
    A matching $M$ is popular if and only if the followng conditions hold:
    \begin{itemize}
        \item every edge of $M$ is from $G'$,
        \item every $a\in A$ who has degree two in $G'$ is matched in $M$,
        \item every house $h\in H$ that can be saturated with admirers of $h$ is saturated with admirers of $h$ only and
        \item  if a house can have all its admirers, than it has all its admirers matched to it.  
    \end{itemize}
\end{theorem}

We note that the first three conditions also imply the fourth one. This is because the admirers of such a house $h$ have degree two and both their edges are going to $h$ in $G'$.
This implies the following.
\begin{corollary}
\label{cor:popperf}
    There is a applicant-perfect popular matching if and only if there is a matching $M$ in $G'$ that matches all applicants and  every house $h\in H$ that can be saturated with admirers of $h$ is saturated with admirers of $h$ only.  
\end{corollary}
Now we show the key lemma for our algorithm.
\begin{lemma}
\label{lemma:popperf}
Let $I$ be an instance of the capacitated house allocation problem. Let $x$ be the size of a maximum size matching $M$ in $G'$ that satisfies that every house $h\in H$ that can be saturated with admirers of $h$ is saturated with admirers of $h$ only and every house that has at most as many admirers as its capacity has all its admirers matched to it. Then, if we increase the capacity of a house by one, the maximum size of such a matching can only increase by one. 
\end{lemma}
\begin{proof}
    Clearly, a matching $M$ satisfying the two conditions always exits (we can find one by greedily assigning applicants to first houses, if there are free places left there), so $x>0$.

    We first note that the statement is not trivial, because if we increase the capacity of a house by one, then the edges of the graph $G'$ may change too.  Observe, that $f(a)$ can never change, only $s(a)$ can.

    Let the house whose capacity got increased by one be $h$.
    Clearly, if the edges of $G'$ do not change then the statement is trivial.

    Suppose the edges of $G'$ change and denote the new graph by $G''$. This can arise for two reasons. 
    
    Firstly, it can arise because $h$ had a strictly smaller capacity than the number of admirers that it has, but after the increase $h$ has capacity equaling the number of admirers it has. In this case the edges of $G'$ change in a way such that for every admirer $a$ of $h$, we delete the edge $(a,s(a))$ (if it existed) and add a parallel edge $(a,f(a))$. This is because for those applicants, $s(a)$ changes to $f(a)$. The adjacent edges of other applicants do not change. Suppose there is a maximum size matching $M'$ that has size at least $x+2$ and satisfies the conditions (with respect to the new capacities) of the lemma. We can assume that $M'$ does not contain any edges of the form $(a,s(a))$, whenever $s(a)=f(a)$, as we can take $(a,f(a))$ instead. Then, we delete an edge of $M'$ that is adjacent to $h$ (there must be at least one). Then, the obtained matching $M''$ is feasible with the original capacites too in $G''$, and has size at least $x+1>|M|=x$. Also, each edge of $M''$ is in $G'$ too (only edges adjacent to admirers of $h$ got deleted, and each new edge already had a parallel edge in $G'$). Finally, $M''$ satisfies the conditions of the Lemma, because $M'$ satisfied it (as $h$ still can be saturated with admirers only after the increase). Hence, we got a contradiction, as $M''$ would be a larger matching in $G'$ than $M$ satisfying the conditions.

    Secondly, the change in $G'$ could arise because house $h$ had capacity being equal to the number of its admirers, which got increased by one. In this case, there may be some applicants $a\in A$, such that $s(a)$ was worse for $a$ then $h$ and for these applicants, $s(a)$ changes to $h$. For the other applicants $s(a)$ does not change. Suppose that there is a matching $M'$ in $G''$ that satisfies the conditions (with respect to the new capacities) and has size at least $x+2$. Because it satiesfies the conditions, at most one applicant $a\in A$ can be matched to $h$, for whom $h$ is not $f(a)$, hence at most one new edge may be in $M'$. Let $M''$ be the matching we obtain from $M'$ by deleting the edge from $h$ to a non-admirer applicant, if there is any, otherwise $M''=M'$. Then, $|M''|\ge x+1>|M|$. Also, $M''$ is feasible with respect to the original capacities, every edge of $M''$ is in $G'$ too and $M''$ satisfies the conditions (because $M'$ satisfied it and all admirers of $h$ are matched to $h$ in $M'$), which is a contradiction again. 
\end{proof}

\subsection{\minsum\ version}

Now we are ready to show that \sumpop\ can be solved if only increases are allowed.
\begin{theorem}
    \sumpop\ is polynomial-time solvable, if only increases are allowed.
\end{theorem}
\begin{proof}
    We provide a polynomial-time algorithm to solve it. The algorithm proceeds as follows. In the first phase, we start by creating the graph $G'$ and finding a maximum size matching $M'$ that satisfies the conditions of \cref{lemma:popperf}. We can do this as follows. First, for each house $h$ that can admit all its admirers, we match all its admirers to it and fix these edges. Then, for each house $h$ that can be saturated with admirers only, we weigh the edges that connect $h$ to an admirer with weight 1, and all other edges with weight 0 and find a maximum weight maximum size (capacitated) matching (i.e. a matching with maximum weight among the maximum size matchings). It is easy to see that with this weight function, the maximum weight matchings are exactly the ones that saturate each such house $h$ with admirers only.

    In the second phase, for each applicant left alone in $M'$, we match her to her best house and increase the capacities accordingly.

    Let $|M'|=x.$
    By \cref{cor:popperf} and \cref{lemma:popperf}, it holds that the optimum is at least $|A|-x$. It is also clear that the algorithm increases the capacities by exactly this amount.

    Hence we only need to show that the output $M$ is popular. For this, we show that all four conditions from \cref{thm:popperf} hold.
    Let $G''$ be the auxilary graph we obtain with the final capacities. 

    First of all, every edge of $M$ is included in $G''$, because if an applicant is matched to $s(a)\ne f(a)$, then they were matched in the first phase, and for any house $h\succ_a s(a)$ its capacity did not increase to more than the number of admirers it has (such a house must had as many admirers as its capacity matched to it in $M'$ and only received more admirers in the second step), hence $s(a)$ did not change in the second phase.
    Clearly, every applicant with degree two in $G''$ is matched. If a house can be saturated with admirers only, then it could be saturated with admirers only with respect to the original capacities too, so the third condition remains true (we only assigned more admirers to each house in the second step). Finally, if a house $h$ can have all its admirers matched to it, then either it could accept all its admirers originally, that are still all at $h$ or $h$ increased capacity such that now it can accept all its admirers. In that case, there was an admirer not at $h$, so $h$ must have been saturated with admirers only, and because we increased $h$'s capacity to be able to accomodate all admirers, by the construction, all of them must have been matched to $h$, as desired. This concludes that $M$ is popular.
\end{proof}

Sadly, if decreases are also allowed, then \sumpop\ becomes NP-hard.

First we provide an example to illustrate that allowing decreases may be greatly beneficial. 
\begin{example}
We have three houses $h_1$ with capacity 1 and $h_2$ with capacity 2 and $h_3$ with capacity $n+1$.
We have applicants $a_1,\dots,a_{n+2}$ with preference $h_1\succ h_2\succ h_3$ and an applicant $b$ with preference $h_2\succ h_1$. Suppose that we can only increase the capacities. Then, at least $n$ increase is necessary: otherwise $h_1$ and $h_2$ have combined capacity at most $n+2$. Therefore, one $a_i$, say $a_3$ must be at $h_3$. One $a_i$, say $a_2$ must be at $h_2$ too, otherwise the one at $h_3$ would envy a free spot which cannot happen in a popular matching. Similarly, at least one applicant must be at $h_1$. Hence, $a_3$ could go to $h_2$ in $a_2$'s place, and $a_2$ could go to $h_1$ by replacing an applicant there, so two applicants would improve and only one would disimprove, contradiction. 

However, if we can decrease the capacities, then decreasing $h_2$'s capacity to 1 is enough: $M=\{ (a_1,h_1),(b,h_2),(a_i,h_3)\mid i=2,\dots,n+2\}$ is popular and perfect as it satisfies the properties of \cref{thm:popperf}.
\end{example}

We proceed to our hardness result.

\begin{theorem}
\sumpop\ is NP complete if both increases and decreases are allowed.
\end{theorem}
\begin{proof}
As shown by \cite{manlove2006popular}, it is possible to decide whether there is a complete popular matching given some fixed capacities, implying that the problem is in NP.

We provide a reduction from \tdm. Let $I$ be an instance of \tdm. Let $A\cup B\cup C=\{ e_1,\dots,e_{3\hat{n}}\}$.

We create an instance $I'$ of \sumpop\ as follows.
\begin{itemize}
    \item [--] For each element $e_i\in A\cup B\cup C$, we add a house $e_i$,
    \item [--] For each set $S_j$, $j\in [3\hat{n}]$ we add four houses $t_j,p_j,q_j,x_j$ and 6 applicants $s_j^1,s_j^2,s_j^3,a_j,p_j'$ and $q_j'$.
\end{itemize}
The capacities are the following: $\quota [e_i]=1$, $\quota [t_j]=3$, $\quota [p_j]=2$, $\quota [q_j]=1$ for each $i,j\in [3\hat{n}]$. 
The preferences are described in the following table. Let $S_j=\{ e_{j_1},e_{j_2},e_{j_3} \}$, $j_1<j_2<j_3$.

\begin{center}
\begin{tabular}{ll}
    $s_j^1:$ & $e_{j_1}\succ p_j\succ t_j$ \\
    $s_j^2:$ & $e_{j_2}\succ p_j\succ t_j$\\
    $s_j^3:$ & $e_{j_3}\succ p_j\succ t_j$ \\
    $a_j:$ & $q_j\succ p_j\succ x_j$\\
    $q_j':$ & $q_j$  \\
    $p_j':$ & $p_j$
    
\end{tabular}
\end{center}
where $j,i\in [3\hat{n}]$.

\begin{claim}
\label{claim:perfpop4} 
There is a good capacity decrease vector $\pr$ with $|\pr|_1\le 2\hat{n}$ if and only if there is a good capacity change vector $\pr$ with $|\pr|_1\le 2\hat{n}$, if and only if $I$ admits an exact 3-cover.
\end{claim}
\begin{proof}
     First suppose there is an exact 3-cover $S_{j_1},\dots,S_{j_{\hat{n}}}$ in $I$. We show that there is a good capacity decrease vector with $|\pr |_1=2\hat{n}$. Denote $J=\{ j_1,\dots,j_{\hat{n}}\}$. Define $\pr$ as follows. $\pr [e_i]=0$, for each $i\in [3\hat{n}]$, $\pr [t_j]=0$ for each $j\in [3\hat{n}]$, $\pr [p_j] =0$, if $j\in J$, $\pr [p_j]=-1$, if $j\notin J$, $\pr [x_j]=0$, $\pr [q_j]=0$ for each $j\in [3\hat{n}]$. 

    Create a matching $M$ as follows. Add the edges $\{ (s_j^1,e_{j_1}),(s_j^2,e_{j_2}),(s_j^3,e_{j_3}),$ $(a_j,p_j),(p_j',p_j),(q_j',q_j)\}$, if $j\in J$ and the edges $\{ (s_j^1,t_j),(s_j^2,t_j),(s_j^3,t_j),(a_j,x_j), $ $ (p_j',p_j),(q_j',q_j)\}$ if $j\notin J$. 

    Clearly, $M$ is applicant-perfect and feasible with respect to $\quota +\pr$ and $|\pr |_1=2\hat{n}$. 

    We show that $M$ is popular. For each applicant $a_j$, either her best two houses have capacity one and one admirer, or $a_j$ is matched to her second house and her best house has capacity one, filled by an admirer. For the $s_j^1,s_j^2,s_j^3$ applicants, either all three are at their best houses, or their better houses have capacity one, which is filled by an admirer. Hence, they could only improve by replacing an applicant at their first house. It follows that in any matching $M$, the number of improving applicants is at most the number of disimproving ones, hence $M $ is popular. 

    The implication that if there is a good capacity decrease vector with $|\pr |_1\le 2\hat{n}$, then there is a good capacity change vector with $|\pr |_1\le 2\hat{n}$ is trivial.

    For the other direction, suppose that there is a good capacity change vector with $|\pr |_1\le 2\hat{n}$, such that $\quota +\pr$ admits an applicant perfect popular matching $M$. Note that in that matching it must hold that $p_j'$ is at $p_j$ and $q_j'$ is at $q_j$ for any $j$. We first show that whenever it holds that not all of $s_j^1,s_j^2,s_j^3$ are at $e_{j_1},e_{j_2},e_{j_3}$ respectively, then we have to modify the capacities of $p_j$ or $q_j$. Suppose not. Then, the $s_j^{\ell}$ applicant not at $e_{j_{\ell}}$ is either at $t_j$, while $a_j$ is at $p_j$ with $p_j'$ and $q_j'$ is at $q_j$ or she is at $p_j$ with $p_j'$, $a_j$ is at $x_j$ and $q_j'$ is at $q_j$. Neither would be popular. In the first case, $s_j^{\ell}$ could go to $p_j$ in $a_j$'s place and $a_j$ to $q_j$ (replacing $q_j'$) and in the second case, $a_j$ could go to $p_j$ in $s_j^{\ell}$'s place and $s_j^{\ell}$ to $e_{j_{\ell}}$ to obtain a matching that dominates $M$. 

    Let $l$ denote the number of indices $j\in [3\hat{n}]$ such that $s_j^1,s_j^2,s_j^3$ are all at $e_{j_1},e_{j_2},e_{j_3}$ respectively. If $|\pr |_1=\eta_1+\eta_2\le 2\hat{n}$, where $\eta_1$ is the sum of the capacity changes of the $e_i$ houses, then $l\le \hat{n}+\frac{\eta_1}{3}$, as the initial capacity of each $e_i$ house is 1. By our above observation, we get that the sum of capacity changes is at least $(3\hat{n}-l)+\eta_1\ge 2\hat{n} +\frac{2\eta_1}{3}$. As this is at most $2\hat{n}$ and $\eta_1\ge 0$, we get that $\eta_1=0$. 
    
    This implies that $l\le \hat{n}$. As $(3\hat{n}-l)\le 2\hat{n}$, $l=\hat{n}$. Combining this with $\eta_1=0$, we get that these $l$ indices correspond to $l$ sets, whose corresponding $s_j^1,s_j^2,s_j^3$ applicants are all at their first house and these sets must form an exact 3-cover. Therefore, $I$ admits an exact 3-cover.

\end{proof}

\end{proof}

\subsection{\minmax\ version}
Finally, we consider \maxpop. We show that \maxpop\ is NP-complete, even if $\maxcap
=1$ or $2$ depending on whether decreases are allowed. Then, we also show that is it NP-hard to approximate within any constant factor.

\begin{theorem}
    \maxpop\ is NP-complete in both of the following three cases: 1) only increases are allowed and $\maxcap =2$, 2) both increases and decreases are allowed and $\maxcap =1$.
\end{theorem}
\begin{proof}
As before, given a good capacity change or increase vector we can verify it in polynomial-time, so the problem is in NP.

We provide a reduction from \tdm. Let $I$ be an instance of \tdm. Let $A\cup B\cup C=\{ e_1,\dots, e_{3\hat{n}}\}$.

We start by describing the base construction for both reductions. 

\begin{itemize}
    \item [--] For each element $e_i$, we add a house $e_i$ and an applicant $e_i'$,
    \item [--] For each set $S_j$, we add three houses $t_j,p_j,q_j$ and 6 applicants $s_j^1,s_j^2,s_j^3,s_j^4,p_j'$ and $a_j$,
    \item [--] We add a collector house $x$.
\end{itemize}

Let us start with the case when both increases and decreases are allowed. Let $\maxcap =1$.
For each set $S_j$, we further add two applicants $q_j^1,q_j^2$. Define the capacities of the houses as follows: Each house $e_i$ has capacity 1, each house $t_j$ has capacity 4, each house $p_j$ has capacity 2, each house $q_j$ has capacity 1 and house $x$ has capacity $2\hat{n}-1$ initially.
The preferences are described in the following table. Let the $j$-th set be $S_j=\{ e_{j_1},e_{j_2},e_{j_3}\}$, $j_1<j_2<j_3$.

\begin{center}
\begin{tabular}{ll}
    $s_j^1:$ & $e_{j_1}\succ p_j\succ t_j$ \\
    $s_j^2:$ & $e_{j_2}\succ p_j\succ t_j$\\
    $s_j^3:$ & $e_{j_3}\succ p_j\succ t_j$ \\
    $s_j^4:$ & $e_{j_3}\succ p_j\succ t_j$ \\
    $a_j:$ & $q_j\succ p_j\succ x$\\
    $q_j^1:$ & $q_j$ \\
    $q_j^2:$ &$q_j$ \\
    $p_j':$ & $p_j$ \\
    $e_i':$ & $e_i$
    
\end{tabular}
\end{center}
where $j\in [3\hat{n}]$.

\begin{claim}
    \label{claim:popperf1}
    There is a good capacity change vector $\pr$ with $|\pr|_{\infty}\le 1$ if and only if $I$ admits an exact 3-cover. 
\end{claim}
\begin{proof}
    For the first direction, suppose there is an exact 3-cover $S_{j_1},\dots,S_{j_{\hat{n}}}$. Denote $J=\{ j_1,\dots,j_{\hat{n}}\}$. Define $\pr$ as follows. $\pr [e_i]=1$, for each $i\in [3\hat{n}]$, $\pr [t_j]=0$ for each $j\in [3\hat{n}]$, $\pr [p_j] =1$, if $j\in J$, $\pr [p_j]=-1$, if $j\notin J$, $\pr [x]=1$ and $\pr [q_j]=1$ for each $j\in [3\hat{n}]$. 

    Create a matching $M$ as follows. Add the edges $\{ (s_j^1,e_{j_1}),(s_j^2,e_{j_2}),(s_j^3,e_{j_3}),(s_j^4,p_j),$ $(p_j',p_j),(a_j,p_j),(q_j^1,q_j),(q_j^2,q_j)\}$, if $j\in J$ and the edges $\{ (s_j^1,t_j),(s_j^2,t_j),(s_j^3,t_j),(s_j^4,t_j),$ $(p_j',p_j),(a_j,x),(q_j^1,q_j),(q_j^2,q_j)\}$ if $j\notin J$. Finally, we add the edges $\{ (e_i',e_i)\mid i\in [3\hat{n}]\}$.

    Clearly, $M$ is applicant-perfect and $|\pr |_{\infty}=1$. $M$ is feasible, as only $2\hat{n}$ $a_j$ applicants are matched to $x$.

    It only remains to show that $M$ is popular. For each applicant $a_j$,  her better houses are filled with admirers. So when $a_j$ could improve, she would have to replace $q_j^1,q_j^2$ or $p_j'$, who could not improve. Same holds for the $s_j^4$ applicants, their better houses are filled by admirers. For the $s_j^1,s_j^2,s_j^3$ applicants, either all three are at their best houses, or they are at their worst houses, but their better houses are filled with admirers. It follows that in any matching $M'$, the number of improving applicants is at most the number of disimproving ones, hence $M$ is popular. 

    For the other direction suppose that there is a good capacity change vector with $|\pr |_{\infty}=1$, such that $\quota +\pr$ admits an applicant perfect popular matching $M$. As $|\pr |_{\infty}\le 1$, it holds that there are at most $2\hat{n}$ $a_j$ applicants who are assigned to the collector house $x$. Let $J=\{ j_1,\dots,j_l\}$, $l\ge \hat{n}$ be the indices such that $a_{j_i}$ is not at house $x$ in $M$. We claim that for each such $j\in J$, it holds that three of $s_j^1,s_j^2,s_j^3,s_j^4$ are matched to $\{ e_{j_1},e_{j_2},e_{j_3}\}$. As the quota of these houses can be at most two, but $e_i'$ must be at $e_i$, it then follows that $l= \hat{n}$ and the sets corresponding to the indices in $J$ form an exact 3-cover. So suppose for the contrary, that there is a $j\in J$ such that two of $s_j^1,s_j^2,s_j^3,s_j^4$ are not at an $e_i$ house. Observe that house $p_j$ has capacity at most 3 after the changes. If both of them are at $p_j$ with $p_j'$ who must be there ($M$ is perfect), then $a_j$ must not be there. As $M$ is applicant perfect, $q_j^1,q_j^2$ are matched to $q_j$ and by  $|\pr |_{\infty}\le 1$, $q_j$ has capacity 2. Therefore, $a_j$ is at $x$, contradicition. Hence, one of the $s_j^{\ell}$ applicants must be at $t_j$. As $j\in J$, we know that $a_j$ must be matched to $p_j$ along with $p_j'$ (it is  not as $x$, and cannot be at $q_j$). Hence, if this $s_j^{\ell}$ applicant goes to $p_j$ to $a_j$'s place and $a_j$ goes to $q_j$, then 2 applicants improve and only one gets worse, so we would obtain a matching that dominates $M$, contradiction again.
    Therefore, $I$ admits an exact 3-cover.
\end{proof}

For the case when only increases are allowed, we modify the construction as follows. We change the initial capacity of $p_j$ to 1 for all $j$, and change the capacity of $x$ to $2\hat{n}-2$. Finally, instead of one dummy $e_i'$ and two dummy $q_j^1,q_j^2$ applicants, we have two $e_i^1,e_i^2$ (who all consider only $e_i$ acceptable) and three $q_j^1,q_j^2,q_j^3$ (who all consider only $q_j$ acceptable) applicants. Finally, let $\maxcap =2$ in this case.


    

\begin{claim}
\label{claim:perfpop2} 
There is a good capacity increase vector $\pr$ with $|\pr|_{\infty}\le 2$ if and only if $I$ admits an exact 3-cover.
\end{claim}
\begin{proof}
 For the first direction, suppose there is an exact 3-cover $S_{j_1},\dots,S_{j_{\hat{n}}}$. Denote $J=\{ j_1,\dots,j_{\hat{n}}\}$. Define $\pr$ as follows. $\pr [e_i]=2$, for each $i\in [3\hat{n}]$, $\pr [t_j]=0$ for each $j\in [3\hat{n}]$, $\pr [p_j] =2$, if $j\in J$, $\pr [p_j]=0$, if $j\notin J$, $\pr [x]=2$, $\pr [q_j]=2$ for each $j\in [3\hat{n}]$. 

   Create a matching $M$ as follows. Add the edges $\{ (s_j^1,e_{j_1}),(s_j^2,e_{j_2}),(s_j^3,e_{j_3}),$ $(s_j^4,p_j),(a_j,p_j)\}$, if $j\in J$ and the edges $\{ (s_j^1,t_j),(s_j^2,t_j),(s_j^3,t_j),(s_j^4,t_j),(a_j,x)\}$ if $j\notin J$. Finally, match each dummy applicant $e_i^l,q_i^l,p_i'$ to their only houses.

 By similar reasoning as in the previous case, it is straightforward to verify that $M$ is feasible and popular.  


For the other direction suppose that there is a good capacity increase vector with $|\pr |_{\infty}\le 2$, such that $\quota +\pr$ admits an applicant perfect popular matching $M$. As $|\pr |_{\infty}\le 2$, it holds that there are at most $2\hat{n}$ $a_j$ applicants who are assigned to the collector house $x$. Let $J=\{ j_1,\dots,j_l\}$, $l\ge \hat{n}$ be the indices such that $a_{j_i}$ is not at house $x$. We claim that for each such $j\in J$, it holds that three of $s_j^1,s_j^2,s_j^3,s_j^4$ are matched to $\{ e_{j_1},e_{j_2},e_{j_3}\} $. As the quota of these houses can be at most 3 and two of the places must be occupied by $e_i^1,e_i^2$, it then follows that $l= \hat{n}$ and the sets corresponding to the indices in $J$ form an exact 3-cover. 
    
So suppose for the contrary, that there is a $j\in J$ such that two of $s_j^1,s_j^2,s_j^3,s_j^4$ are not at a house $e_i$. Observe that house $p_j$ has capacity at most 3 after the changes and one place filled with $p_j'$. If both of them are at $p_j$, then $a_j$ must not be there. As $M$ is applicant perfect, $q_j^1,q_j^2,q_j^3$ are matched to $q_j$ and by  $|\pr |_{\infty}\le 2$, $q_j$ has capacity exactly 3. Hence, $a_j$ is at $x$, contradicting $j\in J$.  Therefore, one of the $s_j^{\ell}$ applicants must be at $t_j$. As $j\in J$, we know that $a_j$ must be matched to $p_j$. Hence, if this $s_j^{\ell}$ applicant goes to $p_j$ and $a_j$ goes to $q_j$ (replacing $q_j^1$), 2 applicants improve and only one gets worse, so we would obtain a matching that dominates $M$, contradiction again.
    Therefore, $I$ admits an exact 3-cover.

\end{proof}

The statements of the theorem follows from these two claims.
\end{proof}

Finally, we show that (the optimization version of) \maxpop\ cannot be approximated within any constant factor, if $P\ne NP$.

\begin{theorem}
    There is no polynomial-time $d$-approximation for \maxpop, for any constant $d$, if $P\ne NP$. This holds both when only increases are allowed and when increases and decreases are both allowed.
\end{theorem}
\begin{proof}
    We use the fact that \scov\ does not admit a constant factor polynomial time approximation, if $P\ne NP$ \cite{feige1996threshold}.

Let $d$ be an arbitrary constant.
We show that if there would be an algorithm $\mathcal{A}'$ for \maxpop\ that given an instance $I'$ of \maxpop, can find a good capacity change vector $\pr$ with $|\pr |_{\infty}\le d\cdot  OPT(I')$, then we could construct an algorithm $\mathcal{A}$ that given an \scov\ instance $I$, can find a set cover using at most $2d\cdot OPT(I)$ sets, which would give a constant factor approximation for \scov.
Let $I=(\mathcal{S},E)$ be a instance of \scov\ and let $k=OPT(I)<\infty$. Create an instance $I'$ of \maxpop\ as follows. Let $N=\hat{n}^2\hat{m}$ and $l_j=|S_j|$ for $j\in [\hat{m}]$. We can suppose that $\hat{n}>k$ holds, as otherwise it is easy to find a set cover with $\hat{n}$ sets in a greedy way. We can also suppose that $d<\hat{n}$ also holds, as $d$ is a constant.

    \begin{itemize}
        \item [--] For each set $S_j$, create houses $s_j^1,\dots,s_j^{l_j}$ with capacity 1 and dummy applicants $t_j^1,\dots,t_j^{l_j}$. Also, create a house $w_j$ with capacity $N$ and applicants $a_j^1,\dots,a_j^N$.
        \item[--] For each element $e_i\in E$, we create an applicant $e_i$.
        \item[--] We create a house $f$ with capacity 1 and a dummy applicant $f'$ and we create houses $x^1,\dots,x^N$ with capacity 1 and dummy applicants $y^1,\dots,y^N$.
    \end{itemize}

We set the preferences as follows.
For an element $e_i\in E$, let $\mathcal{S}_i$ be the set of those $s_j^l$ houses, such that $e_i\in S_j$, and $e_i$ is the $l$-th smallest index element in $S_j$.

\begin{center}
\begin{tabular}{ll}
    $e_i:$ & $f\succ  \seqq{\mathcal{S}_i}$ \\
    $f':$ & $f$ \\
    $a_j^l:$ & $x^l\succ s_j^1\succ \dots \succ s_j^{l_j}\succ w_j$ \\
    $t_j^{l}:$ & $s_j^{l}$ \\
    $y^l:$ & $x^l$,
    
\end{tabular}
\end{center}
for $j\in [\hat{m}], i\in [\hat{n}], l\in [N]$.

We first show that $OPT(I')\le k$. Let $S_{j_1},\dots,S_{j_k}$ be a set cover with $k $ sets in $I$. Create a capacity increase vector $\pr$ as follows: set $\pr [s_j^l]=1$ for each $j\in \{ j_1,\dots,j_k\}$ and each $l$. Also, set $\pr [x^l]=k$ for each $l\in [N]$. For every other house, we do not change the capacities. 

Clearly, $|\pr |_{\infty} =k$. Let $M$ be the following matching. We assign each dummy applicant $y^l$ to $x^l$, $t_j^l$ to $s_j^l$ and $f'$ to $f$. Then, for each applicant $e_i$, we assign her to the best house in $\seqq{\mathcal{S}_i}$ that increased its capacity. Clearly, there must be such a house, as the $k$ set covered every element. Also, each $s_j^l$ house is acceptable to only one $e_i$ applicant (the one with the $l$-th smallest index in $S_j$), so there are no collisions. Finally, for an $a_j^l$ applicant, if $j\in \{ j_1,\dots,j_k\}$, then we assign her to $x^l$, otherwise we assign her to $w_j$. As we used exactly $k$ sets, all $x^l$ house gets exactly $k+1$ applicants, all of them being admirers. $M$ is clearly perfect. We claim that $M$ is popular. Indeed, for any applicant, it holds that any house she considers better than the one she is assigned to is filled with admirers. Hence, no set of applicants could improve, without making at least as many applicants worse off. 

Now, suppose that we have an algorithm $\mathcal{A}'$ that finds a good capacity change vector $\pr$ with $|\pr |_{\infty}\le d\cdot OPT(I')\le dk$. Clearly, no $s_j^l$, $x^l$ house nor $f$ can decrease its capacity, otherwise their dummy applicants cannot be assigned. Let $M$ be a perfect popular matching with the new capacities $\quota +\pr$. Then, at most $dk$ of the $e_i$ applicants can be at $f$, so at least $\hat{n}-dk$ of them are assigned to $s_j^l$ houses. Let $J=\{ j_1,\dots,j_{\ell}\}$ be the set of those $j\in [\hat{m}]$ indices, such that there is an $s_j^l$ house that increased its capacity. Fix a $j\in J$ and let $s_j^p$ be the house with smallest upper index who increased its capacity. We claim that all of $a_j^1,\dots,a_j^N$ are either at $s_j^p$ or at their best houses. If $a_j^1$ is not at $x^1$ or $s_j^p$, then $a_j^1$ can only be at a worse house for her than $s_j^p$. Hence, we could create a matching $M'$ dominating $M$ as follows. Let $a_j^1$ go to $s_j^p$ and let someone (other than $t_j^p$) from $s_j^p$ switch to her best house. There must be such a student at $s_j^p$, as $s_j^p$ increased its capacity and there can be no free places at an envied house in $M$. Hence, we get that $M$ would not be popular, contradiction. Similarly, we get that all of $a_j^1,\dots,a_j^N$ are either at $s_j^p$ or their best house. Also, at most $dk$ of them can be at $s_j^p$. From this, we get that the number of places at the $x^l$ houses that $a_j^l$ applicants fill is at least $|J|\cdot (N-dk)$. As $|\pr |_{\infty}\le dk$, the number of such places is at most $dkN$. Hence, $dkN\ge |J|\cdot (N-dk)$, so $|J|\le \frac{dk(N+|J|)}{N}\le dk+\frac{dk|J|}{\hat{n}^2\hat{m}}< dk+\frac{\hat{n}^2\hat{m}}{\hat{n}^2\hat{m}}$, by using that $N=\hat{n}^2\hat{m}$, $|J|\le \hat{m}$ and $dk< \hat{n}^2$. 
Therefore, we can conclude that there are at most $dk$ such sets $S_j\in \mathcal{S}$, such that a corresponding $s_j^l$ house has increased capacity. This implies that the sets with indices from $J$ cover at least $\hat{n}-dk$ elements. Clearly, from these sets, we can construct a set cover using at most $2dk$ sets, by adding at most $dk$ new sets, each covering a new element from the uncovered $dk$ ones (this is possible, as there is a set covering each element). Hence, we could construct a $2d$-approximation algorihm for \scov, a contradiction.

\end{proof}

\section{Acknowledgements}
Gergely Csáji  acknowledges the financial support by the Hungarian Scientific Research Fund, OTKA, Grant No. K143858, by the Momentum Grant of the Hungarian Academy of Sciences, grant number 2021-1/2021 and by the Ministry of Culture and Innovation of Hungary from the National Research, Development and Innovation fund, financed under the KDP-2023 funding scheme (grant number C2258525).
\clearpage

\bibliographystyle{ACM-Reference-Format} 

\bibliography{cit}

\clearpage

\end{document}